\documentclass[11pt]{article}
\usepackage{hyperref}

\usepackage{tikz-cd}
\usepackage{xcolor}
\usepackage{amsmath,amsthm,amsfonts,amssymb,amsopn,amscd,bm,cancel} 
\usepackage{color}
\usepackage{slashed}

\usepackage{soul}

\def\qed{{\unskip\nobreak\hfil\penalty50
\hskip2em\hbox{}\nobreak\hfil$\square$
\parfillskip=0pt \finalhyphendemerits=0\par}\medskip}
\def\proof{\trivlist \item[\hskip \labelsep{\bf Proof.\ }]}
\def\endproof{\null\hfill\qed\endtrivlist\noindent}
\def\eproof{\null\hfill\qed\endtrivlist\noindent}

\def\tilde{\widetilde}

\def\a{\alpha}

\def\Ga{\Gamma}

\def\l{\lambda}

\def\Om{\Omega}

\def\A{{\cal A}}

\def\E{{\cal E}}

\def\cF{{\cal F}}

\def\H{{\cal H}}

\def\S{{\cal S}}

\def\l{{\lambda}}

\def\PSL{{{\rm PSL}(2,\mathbb R)}}

\def\S2{S^{1(2)}}
\def\Poi{{\cal P}_+^\uparrow}

\def\RR{\mathbb R}

\newtheorem{theorem}{Theorem}[section]
\newtheorem{lemma}[theorem]{Lemma}

\newtheorem{corollary}[theorem]{Corollary}

\newtheorem{proposition}[theorem]{Proposition}

\theoremstyle{remark} 

\newcommand{\ben}{\begin{equation}}
\newcommand{\een}{\end{equation}}

\newcommand{\bthm}{\begin{theorem}}
\newcommand{\ethm}{\end{theorem}}

\newcommand{\bprop}{\begin{proposition}}
\newcommand{\eprop}{\end{proposition}}
\newcommand{\bcor}{\begin{corollary}}
\newcommand{\ecor}{\end{corollary}}
\newcommand{\blem}{\begin{lemma}}
\newcommand{\elem}{\end{lemma}}

\def\PSL{PSU(1,1)}

\def\CC{{\mathbb C}}

\def\SL2{{{\rm SL}(2,\bR)}}

\def\PSL2{{{\rm PSL}(2,\bR)}}

\def\U1{{{\rm V}(1)}}
\def\SU2{{{\rm SV}(2)}}

\def\SU{{{\rm SU}}}

\def\A{{\mathcal A}}

\def\H{{\mathcal H}}

\def\cI{{\mathcal I}}
\def\cK{{\mathcal K}}

\def\S{{\mathcal S}}
\def\T{{\mathcal T}}

\def\bx{{\bf x}}
\def\by{{\bf y}}
\def\bp{{\bf p}}
\def\bq{{\bf q}}

\def\cE{{\cal E}}
\def\cF{{\cal F}}

\def\cH{{\cal H}}
\def\cI{{\cal I}}

\def\cK{{\cal K}}

\def\cM{{\cal M}}

\def\bC{{\mathbb C}}

\def\bN{{\mathbb N}}

\def\bR{{\mathbb R}}
\def\RR{{\mathbb R}}
\def\bS{{\mathbb S}}

\def\bZ{{\mathbb Z}}

\def\a{\alpha}

\def\l{\lambda}       \def\L{\Lambda}

\def\c{\chi}

\def\fL{{\mathfrak L}}

\def\S{{S(\RR^d)}}
\def\ov{\overline}
\def\supp{{\text{supp}\,}}

\title{\Huge{Bekenstein's bound for wave packets}
} 
\author{ {\sc Stefan Hollands}\\
Institut für Theoretische Physik, Universität Leipzig \\
Brüderstrasse 16, 04103 Leipzig, Germany\\
Max Planck Institute for Mathematics in Sciences (MiS)\\
Inselstra{\ss}e 22, 04103 Leipzig, Germany
\\
\\
{\sc Roberto  Longo and Gerardo Morsella}\\
Dipartimento di Matematica,
Tor Vergata Universit\`a di Roma\\
Via della Ricerca Scientifica, 1, I-00133 Roma, Italy
}
\date{}
\begin{document}
\maketitle
\begin{abstract}
Let $B$ be a spatial region of width $2R$ and $\Phi$ a Klein-Gordon wave packet localized in $B$ at time zero. 
We show the inequality $S \leq 2\pi R E$; here, $S$ is the entropy of $\Phi$ contained in a region $B$, and $E$ is the energy content of $\Phi$ within $B$. 
We consider a wider setting and formulate a variational problem aimed at minimizing our bound when $\Phi$ is not localized in $B$. 
Our inequality holds in more generality in the framework of local, Poincaré covariant nets of standard subspaces and is related to the Bekenstein inequality.
We point out a general bound that is compatible with the recent numerical computations by Bostelmann, Cadamuro, and Minz concerning the one-particle modular Hamiltonian of a scalar massive quantum Klein-Gordon field.
We also provide a version of the entropy balance and ant formulas for wave packets.
\end{abstract}

\newpage

\section{Introduction}
This note concerns the linear space of Klein-Gordon wave packets and, more generally, local nets of standard subspaces of a Hilbert space. It has been written with two motivations. 

Our first aim is to provide general bounds for the massive modular Hamiltonian associated with a spatial ball in the free, scalar Quantum Field Theory. The explicit description for the modular Hamiltonian of a ball is known only in the massless case \cite{LM23}; in this case, the modular group is associated with a geometrical action by spacetime conformal transformations \cite{HL82}. The massive case does not correspond to any geometric action; one might hope to find out a perturbation formula from the massless case, but any attempt to this end has so far failed, see \cite{LM23}. 
Based on these premises, Bostelmann, Cadamuro, and Minz have recently undertaken an interesting, detailed numerical analysis \cite{BCM23}. In the two-by-two matrix description of Cauchy data, the generator of the modular group of the unit ball is off-diagonal with two terms: $L$ and $M$ (up to a constant), see formula \eqref{LM} below; $M$ and $L$ contain the same information as $M^* = - L$ in the underlying inner product; in particular, the $M$ term has 
been plotted in \cite{BCM23} as the mass $m$ varies. 
In the massless case, $M$ is (the multiplication operator by) a parabolic function, $M = \frac12(1 -x^2)$, thus $M(\pm1) =0$ and $M(0) =1/2$; as the mass increases, the plot of $M$ resembles a higher parabola shape; still $M(\pm 1) = 0$ but $M(0)$ increases to 1, cf. \cite[Figures 7 \& 8]{BCM23}. 
We shall indeed see in Corollary \ref{LMb} that the mass independent bound $M \leq 1$ holds as a consequence of the general bound derived in \cite{L25}. 

Our second motivation comes from the local entropy of a wave packet that was defined in \cite{L19,CLR20}, see also \cite{CLRR22,L23}. 
The local entropy measures the information content of a wave packet in a given region. Now, the Bekenstein bound in Quantum Field Theory \cite{Ca08} suggests that a local entropy/energy bound is to be satisfied. 
The recent Bekenstein-type inequality in Quantum Field Theory \cite{L25,HL} will show that a local entropy/energy bound for a wave packet or, more generally, for a local net of standard subspaces of a Hilbert space is satisfied, where the local energy is defined more precisely below in terms of the stress energy tensor of the wave. This bound is presented in this paper and was actually the starting point for the derivation of QFT Bekenstein-type bound \cite{L25}. 

Let $B$ be a spatial region of width $2R$ and $\Phi$ a Klein-Gordon wave packet, that is, $\Phi$ is a smooth solution of the Klein-Gordon equation $(\square + m^2)\Phi =0$. If the Cauchy data $(f,g)$ of $\Phi$ are supported in $B$, we will derive the bound
\[
S(\Phi | B) \leq 2 \pi R \int_{B} {T_{00}}_\Phi(x)|_{x_0 = 0} d\bx\, ,
\]
with $S(\Phi | B)$ the entropy of $\Phi$ in $B$ and $T_{00}$ the stress-energy tensor. If $f,g$ are no longer supported in $B$, the above inequality does not hold, due to boundary contributions to the entropy and the energy. In order to derive a bound also in this situation, we formulate a variational problem that has its own analytical interest; its explicit solution goes beyond the frame of this paper. Here, we just note in Section \ref{VP} some properties and solutions in particlar cases. We also briefly discuss the extension of such bounds to the fermionic case and obtain a version of the above inequality for one particle states of the Majorana field localized in $B$.

In Section \ref{eba} we provide a version the entropy balance formula \cite{W17,CF18} and ant formula \cite{CF18,HL25} for wave packets: the first one follows at once directly in our context; the second one cannot be inferred from the ant formula for free QFT nets of von Neumann algebras. Indeed, in that context, the version proven here would involve an infimum taken only over coherent states, hence it is a stronger statement.

\section{Preliminaries}
\subsection{The entropy operator}
We recall the setup of standard subspaces. Let $\H$ be a complex Hilbert space and $H\subset\H$ a {\it standard subspace}, namely $H$ is a closed, real linear subspace of $\H$ such that $H\cap iH = \{0\}$ and $\ov{H +iH} = \H$. The {\it modular operator} $\Delta_H$ is a positive, non-singular operator on $\H$, and the {\it modular conjugation} $J_H$ is a anti-unitary involution on $\H$ defined by the polar decomposition $S_H = J_H \Delta_H^{1/2}$ of the anti-linear, densely defined involution $S_H : H + iH\to H + iH$, $S(\Phi + i\Psi) = \Phi - i\Psi$; they satisfy the fundamental relations
\[
\Delta_H^{is}H = H\,, \ s\in\RR\, ,\quad J_H H = H'\, ,
\]
where $H' := \{\Phi\in\H: \Im(\Phi,\Psi) = 0,  \Psi\in H\}$ is the {\it symplectic complement} of $H$. We refer to \cite{L08} for details. 

The {\it entropy operator} $\E_H$ associated with $H$ is a real linear, positive operator on $\H$. $\E_H$ associated with any closed, real linear subspace $H$ of $\H$. Here, we recall the definition of $\E_H$ with $H$ a standard subspace which is {\it factorial}, namely $H\cap H' = \{0\}$ (which is the case of interest in the following):
\[
\E_H =: i P_H i \log\Delta_H\, ,
\] 
more precisely, $\E_H$ is the closure of the operator in the right-hand side of the above formula. Here, $P_H$ is the {\it cutting projection}, namely the real linear operator on $\H$ with domain $H +H'$ defined by $P_H(\Phi + \Phi') = \Phi$, with $\Phi\in H$, $\Phi'\in H'$. Note that $P_H$ is densely defined as $H$ is standard and
well-defined as $H$ is factorial. 

The entropy operator is positive, selfadjoint (as real operator on the real Hilbert space $\H$ with scalar product $\Re(\cdot,\cdot)$) and $\E_H \leq \E_K$ if $H\subset K$. 
The {\it entropy of a vector} $\Phi$ with respect to $H$ is then defined by
\ben\label{Ev}
S(\Phi | H) := \Re (\Phi, \E_H \Phi)
\een
in the quadratic form sense, thus $S(\Phi | H) = \|\E^{1/2}_H \Phi\|^2$ if $\Phi \in D(\E^{1/2}_H)$, and $S(\Phi | H) = +\infty$ otherwise. 
$S(\Phi | H)$ satisfies natural properties like {\it positivity}: $S(\Phi | H)\geq 0$, and {\it monotonicity}: $S(\Phi | H)\leq S(\Phi | K)$ if $H\subset K$. 
See \cite{L19, CLR20, L23} as references for this setup. 

We note that, if $H$ is any standard subspace of a Hilbert space $\H$, the entropy operators with respect to $H$ and $H'$ are related by
\ben\label{EE'}
\E_H - \bar\E_{H} = - \log\Delta_H
\een
(on the intersections of the domains), with $\bar \E_{H} = \E_{H'}$, as $\E_{H'}$ is the closure of $iP_{H'}i  \log\Delta_{H'} =  - i( 1 - P_{H}) i \log\Delta_{H}$.  

\subsection{Bekenstein inequality for nets of standard subspaces}\label{SB}
Let $\H$ be a complex Hilbert space. A Poincaré covariant {\it net of standard subspaces} (see \cite{L08}) of $\H$ is a map $H: O\mapsto H(O)$ that associates a standard subspace $H(O)$ of $\H$ to each double cone $O\subset \RR^{1+d}$, $d \geq 1$, such that

$\bullet$ $H(O_1) \subset H(O_2)$ if $O_1 \subset O_2$ (isotony),

$\bullet$ there exists a unitary, positive energy representation $U$ of the proper, $1 +d$-dimensional Poincaré group $\Poi$ on $\H$ such that $U(g)H(O) = H(gO)$, for all $g\in\Poi$, $O$ double cone (Poincaré covariance),

$\bullet$ $H(O_1) \subset H(O_2)'$ if $O_1 \subset O'_2$ (locality),

$\bullet$  $\Delta^{-is}_{H(W)} = U\big(\L_W(2\pi s)\big)$, $s\in\RR$, if $W\subset \RR^{1+d}$ is a wedge (Bisognano-Wichmann property);

\noindent
here $H(W)$ is the closed real linear span of the $H(O)$ with $O\subset W$ and $\L_W$ is the one-parameter group pure Lorentz transformations preserving $W$ (boosts). The Bisognano-Wichmann property holds automatically if $U$ is irreducible \cite{M18}. 

Recall that every unitary, positive energy representation $U$ of $\Poi$ as above gives rise to a net standard subspaces as above, provided $U$ has finite spin \cite{BGL02}. If $U$ has infinite spin, then $H(O) =\CC$ if $O$ is a double cone, but $H(W)$ is defined if $W$ is a wedge \cite{LMR16}. 

If $B\subset \RR^d$ is a spatial (time zero) region, we set $H(B) = H(\hat B)$, where $\hat B\subset \RR^{1 +d}$ is the causal envelope of $B$. 
We shall say that $B$ has half-{\it width} $R\geq 0$ if $R$ is the smallest number such that $B$ lies between two parallel hyperplanes with distance $2R$. In this case, by applying a rigid motion, $B$ can be mapped onto a subset of the strip  $\{ -R \leq x_1 \leq R\}$. Clearly, a ball of radius $R$ has a half-width equal to $R$. 

The next proposition provides the version of the Bekenstein-type inequality for local, Poincaré covariant, net of standard subspaces \cite{L25} (more generally, it holds for any local net of real Hilbert spaces with energy-momentum spectrum condition  
\bprop\label{beka}
If $\Phi\in H(B)$, then
\ben\label{Ineq}
S(\Phi | B)  \leq 2\pi R\, (\Phi, P\Phi) \, ,
\een
with $P$ the Hamiltonian (the generator of the time-translation one-parameter unitary group) and $R$ is the half-width of $B$. 
\eprop
\proof
The proposition is a corollary of the inequality \cite[Cor. 2.11]{L25} (following from \cite{BDL07}):
\ben\label{In}
-\log\Delta_{H(B)} \leq 2\pi R\, P \, ,
\een
and the fact that
\[
S(\Phi | B) = - (\Phi, \log\Delta_{H(B)}\Phi)\, .
\]
if $\Phi\in H(B)$. 
\eproof
Note that both sides of equations \eqref{Ineq}, \eqref{In}, and the operator order are defined in the quadratic form sense, see \cite{L25}.

\subsection{The wave Hilbert space}\label{SectWH}
We now recall the concrete construction of the net of standard subspaces associated with the unitary, scalar representation of $\Poi$ of mass $m\geq 0$,
see \cite{LM23} for example. 

Denote by $\S$ the real Schwarz space and fix $m\geq 0$.\footnote{In this section, we assume that $d \geq 2$ if $m=0$.}
As is known, if $f, g\in\S$ there is a unique smooth real function $\Phi(x_0,\bx)$ on $\mathbb R^{1+d}$ which is a solution of the Klein-Gordon equation 
\[
(\square + m^2)\Phi = 0\, ,
\]
$\square \Phi \equiv\partial_{x_0}^2 \Phi - \nabla^2_{\bx}\Phi $,
(a {\it wave packet} or, briefly, a wave) with Cauchy data $\Phi |_{x_0 =0} = f$, $\partial_{x_0}\Phi |_{x_0 =0} = g$. 
We set $\Phi=  \langle f,g\rangle$ and denote
by $\T$ the real linear space of these $\Phi$'s. The Poincaré group $\Poi$ naturally acts linearly on $\T$ and there is a natural, Poincaré invariant, real scalar product on $\T$; the completion of $\T$ then gives a real Hilbert space; moreover, this real Hilbert space is provided with the Poincaré invariant complex structure (i.e., an isometry $i$ such that $i^2 = -1$)
described momentarily, so we obtain a complex Hilbert space $\H_m$. Then $\H_m$ carries a unitary, positive energy representation of $\Poi$, the irreducible, scalar representation of mass $m$.  

We will often use the identification
\ben\label{st}
\S\oplus\S \longleftrightarrow \T\, , \qquad \ f\oplus g \longleftrightarrow  \langle f,g\rangle\, ,
\een  
so we may deal directly with $\S\oplus\S$ rather than $\T$.
Let $H_{m,\pm}$ be the real Hilbert space of tempered distributions $f \in \S'$ such that $\hat f$ is a Borel function with
\ben\label{H12}
 \|f \|^2_{\pm} = \int_{\RR^d} {(|\bp|^2 + m^2)^{\pm\frac12}} | \hat f(\bp)|^2 d\bp< +\infty\,  .
\een
$\S$ is dense in $H_{m,\pm}$ if $d>1$. 
As a real Hilbert space, 
$\H_m$
is equal to the real Hilbert space direct sum $H_m = H_{m,+}\oplus H_{m,-}$. 
The complex structure $i$  is given by the isometry determined by
 \ben\label{i}
 i = \begin{bmatrix}0 & \mu^{-1}\\ -\mu& 0\end{bmatrix}
\een
on $\S\oplus\S$,
where $\mu$ is the multiplication by $\sqrt{\bp^2 + m^2}$ in Fourier transform. 

The imaginary part of the scalar product ({\it symplectic form}) is given by
\ben\label{beta2}
\Im(f_1\oplus g_1, f_2\oplus g_2) =
 ( g_1, f_2) - ( f_1, g_2 ) \,  ,
\een
$f_1, f_2, g_1, g_2\in \S$ ($L^2$-scalar product). 

With $B\subset\RR^d$ a ball, we set
\[
H_{m,\pm}(B) = \big\{f_\pm \in \S : {\rm supp}(f_\pm) \subset B\big\}^-\, ,
\]
and the standard subspace $H_m(B)\subset\H_m$ is 
\ben\label{HB}
H_m(B) =: H_{m,+}(B)\oplus H_{m,-}(B)\, .
\een
For a more general region $B$, $H_m(C) = \{\sum H_m(B) :  B\,  {\rm ball}, B\subset C \}^-$. 

\section{Bekenstein's inequality and Klein-Gordon waves}\label{sect1}

With $\H= \H_m$,
the {\it local entropy} of a wave packet $\Phi = \langle f,g\rangle$ is now defined by \eqref{Ev}
\[
S(\Phi| B) = (\Phi, \E_H \Phi) = (\Phi, iP_Hi \log\Delta_H \Phi)\, ,
\]
for any region $B\subset \RR^d$, with $H=H(B)$ defined in \eqref{HB}. 

Recall that \cite{CLR20} 
\ben
\label{SPhi}
S(\Phi | x_1 > a) = 2\pi \int_{x_1 > a} (x_1 - a) {T_{00}}_\Phi(x)|_{x_0 = 0} d\bx
\een
and
\[
S(\Phi | x_1 < a) = 2\pi \int_{x_1 < a} (a - x_1) {T_{00}}_\Phi(x)|_{x_0 = 0} d\bx
\]
with $T_{00}$ the {\it stress-energy tensor} on $\RR^{1 + d}$
\[
 {T_{00}}_\Phi(x)|_{x_0 = 0}  = \frac12(|\nabla f|^2 + m^2 f^2 + g^2) \, .
\]
One can check that 
\[
(\Phi, P\Phi) = \int {T_{00}}_\Phi(0,\bx) d\bx\, ,
\]
therefore, the following concrete Bekenstein inequality follows as a corollary of Proposition \ref{beka}. We give a short, direct proof of it for future needs. 
\bprop\label{bek1}
If ${\rm supp}(f), {\rm supp}(g) \subset \bar B$, then
\ben\label{Ineqbis}
S(\Phi | B)  \leq 2\pi R\,E(\Phi| B) \, ,
\een
with $E(\Phi| B) = \int_B {T_{00}}_\Phi(0,\bx) d\bx$ and
$R$ the half-width of $B$. 
\eprop
\proof
Since the special Euclidean group acts covariantly and unitarily on $\H_m$, by the isometric covariance of the local entropy
we may assume that $B\subset \{ {\bf x}\in \RR^d:-R \leq x_1 \leq R\}$.   

By the monotonicity of the entropy (which follows from~\cite[Thm.\ 4.5]{CLR20} and the monotonicity of relative entropy between coherent states), we have
\begin{align}
S(\Phi | B)  \leq S(\Phi| \{x_1 > -R\}) &= \pi\int_{x_1 > -R}(x_1 + R) (|\nabla f|^2 + m^2 f^2 + g^2) d\bx   \nonumber \\
&= \pi\int_{B} (x_1 + R) (|\nabla f|^2 + m^2 f^2 + g^2)d\bx  \, ; \label{i1}
\end{align}
similarly,
\begin{align}
S(\Phi | B)  \leq S(\Phi |\{ x_1 < R\}) &= \pi\int_{x_1 < R}(- x_1 + R) (|\nabla f|^2 + m^2 f^2 + g^2) d\bx  \nonumber \\
&= \pi\int_{B}(-x_1 + R) (|\nabla f|^2 + m^2 f^2 + g^2)d\bx  \, .\label{i2}
\end{align}
By taking the mean of \eqref{i1} and \eqref{i2}, we then have 
\begin{align}\label{eqM1}
S(\Phi | B)  &\leq \frac\pi2\int_{B} (x_1 + R) (|\nabla f|^2 + m^2 f^2 + g^2)
+ \frac\pi2\int_{B}(-x_1 + R) (|\nabla f|^2 + m^2 f^2 + g^2)d\bx 
\\
&= \pi R\int_{B}  (|\nabla f|^2 + m^2 f^2 + g^2) d\bx 
\\
&= 2\pi R\int_{B}  {T_{00}}_\Phi(x)|_{x_0 = 0} d\bx = 2\pi R \, E(\Phi | B)\, .
\end{align}
\endproof
\bcor
Consider the differential operator $T = -(1 + r^2)\nabla^2 - r\partial_r$ on $L^2 (B)$, with $B\subset \RR^d$ the unit ball, defined on $C^\infty(B)$ with Dirichlet boundary conditions.  Then,
\[
\l_1 \geq d - 1 
\]
with $\l_1$ the lowest eigenvalue of $T$. 
\ecor
\proof
If the smooth function $f$ vanishes on the boundary $\partial B$, then, by integration by parts, \eqref{ineT} in the following is equivalent to
\[
(f, Tf) \geq (d-1) \|f\|^2
\]
$L^2$ scalar product), namely $T \geq d-1$. Thus \eqref{ineT}  holds for all $f\in C^\infty(\bar B)$ vanishing on $\partial B$ iff the spectrum of $T$ is contained in $[d-1, \infty)$, namely iff $\l_1 \geq d-1$. 
\eproof
We end this section by another corollary of \eqref{SPhi}, simply obtained by applying two derivatives $\partial_a^2$:
\begin{corollary}
For any wave packet, we have that
\ben
\label{qdec1}
\partial_a^2 S(\Phi | x_1 > a) = 2\pi \int_{\mathbb{R}^{d-2}}  {T_{00}}_\Phi(x)|_{x_0=0,x_1 = a} d\bx_{\|} \ge 0
\een
where $\bx_{\|} = (x_2, \dots, x_d)$. 
\end{corollary}
The inequality expressed by this corollary may be seen as a version of the quantum dominant energy condition (QDEC) in the present context. 
We will give a  proof of the QDEC for a general Haag-Kastler Quantum Field Theory in a forthcoming work \cite{HLM}. It would be interesting to gain a better understanding of the relation between the QDEC and the Bekenstein bound in such a more general setting.

\subsection{Bounds for the local, scalar, massive modular Hamiltonian}\label{Bmh}
We now consider the case $f=0$ or $g=0$, where $f,g\in C^\infty_0(B)$.  With $\Phi = \langle 0, g\rangle$, \eqref{eqM1} gives
\ben\label{eq2}
-(\Phi, \log\Delta_B\Phi) = S(\langle 0,g\rangle | B)  \leq\pi R\int_{B}   g^2 d\bx \, ;
\een
with $\Phi = \langle f, 0\rangle$, \eqref{eqM1} gives 
\ben\label{eq3}
-(\Phi, \log\Delta_B\Phi) = S(\langle f,0\rangle | B)  \leq\pi R\int_{B}  \big( |\nabla f|^2 + m^2 f^2 \big)d\bx \, ;
\een
Recall that
\ben\label{LM}
-i\log\Delta_B = \pi \begin{bmatrix}0 & M\\ L & 0\end{bmatrix}\, ,
\een
with $i$ given by \eqref{i}, $M : H_{m,-}(B)\to  H_{m,+}(B)$ and $L : H_{m,+}(B)\to  H_{m,-}(B)$ \cite{FG89}, see also \cite{L23}. Thus \eqref{eq2} and \eqref{eq3}  give the estimates with the $L^2$-norm 
\ben\label{MR}
\Im (g, Mg) \leq R\, \| g \|^2_2,
\een
\ben\label{MR2}
\Im (g, Lg) \leq R\, \big\| |\nabla f| \big\|^2_2  + m^2 \big\|  f \big\|^2_2\, .
\een

\bcor\label{LMb}
As real-linear operators $M, L: L^2(B) \to L^2(B)$, we have the bounds
\begin{gather*}
 0\leq M \leq R\, , \\ 0\leq -L\leq -R(\nabla^2 - m^2)
 \end{gather*}
 on $C^\infty_0(B)$. 
\ecor
\proof
Note that $\Im (g, Mg) \geq 0$ by the positivity of the entropy. Moreover,
\eqref{MR} and \eqref{beta2} give
$\int_B gMg d\bx \leq R \| g \|^2_2$ for $g \in C^\infty_0(B)$. Thus
\[
0\leq (g, Mg)_2 \leq R (g,g)_2
\] 
($L^2$ scalar product). On the other hand, \eqref{MR2} and \eqref{beta2} $0\leq - (g, Lg)_2 \leq R \| |\nabla f|\|_2^2 + m^2 \| f \|^2_2$, thus
\[
0\leq -(f, Lf)_2 \leq R\, (f, -\nabla^2 f + m^2 f)_2 
\]
by integration by parts,
so the corollary holds. 
\eproof

\section{Which bound if $\Phi$ is not localized?}\label{VP}
Let us specify the inequality \eqref{Ineqbis} in the case $m = 0$ and $B$ the unit ball. In this case, the entropy of $\Phi = \langle f,g\rangle$ is given by
\[
S(\Phi |B) = \pi\int_B (1-r^2) {T_{00}}_\Phi(x)\big|_{x_0 = 0}d\bx + \pi D\int_B f^2 d\bx \, ,
\]
with ${T_{00}}_\Phi = |\nabla f|^2 + g^2$ and $D = (d-1)/2$.  Therefore, the inequality \eqref{Ineqbis} is equivalent to
\[
\pi\int_B \frac{1-r^2}{2} (|\nabla f|^2 + g^2) d\bx + \pi D\int_B f^2 dx \leq \pi\int_{B} (|\nabla f|^2 + g^2) d\bx  \, ;
\]
in particular, setting $g=0$, we get the inequality
\ben\label{ineT}
 \int_{B} (1 +r^2) |\nabla f|^2 d\bx \geq (d-1)\int_B f^2 d\bx   \, .
\een
Therefore, the inequality \eqref{Ineqbis} does not generally hold if the support of $f$ is not contained in $\bar B$, as can be seen if $f= 1$ on $\bar B$. 

Moreover, we could consider the ``improved'' traceless stress-energy tensor associated to the wave $\Phi$
\[
T_{\mu \nu \Phi}^i := T_{\mu \nu \Phi} - \frac{d-1}{4d}(\partial_\mu \partial_\nu - \eta_{\mu \nu} \Box)(\Phi^2)
\]
($\eta_{\mu \nu}$ the Minkowski metric), which is also conserved and whose second quantization version also generates the translations. We can then compute, by double 
integration by parts, 
\[
\int_B (1-r^2) \nabla^2(\Phi^2) |_{x_0 =0}d\bx = \int_B (1-r^2) \nabla^2(f^2) d\bx =2 \int_{\partial B} f^2 \,d\sigma -2d \int_B f^2 d\bx,
\]
so that the entropy of $\Phi$ in $B$ can also be expressed in terms of the improved energy density by
\[
S(\Phi | B) = \pi \int_B (1-r^2) T_{00\Phi}^i(x)|_{x_0 = 0} d\bx + \pi \frac{D}d\int_{\partial B} f^2 \,d\sigma,
\]
suggesting that the correction to be added to~\eqref{Ineqbis} when $\Phi$ is not localized in $B$ should only depend on the restriction of $f$ to the boundary of $B$.

It is also interesting to observe that for $B$ a half space one has, similarly,
\[
S(\Phi| x_1 > a) = 2\pi \int_{x_1 > a}(x_1-a) T_{00 \Phi}^i(x)|_{x_0 = 0}d\bx + \pi \frac{D}d \int_{x_1 = a} f^2\,d\sigma,
\]
which also holds for $m > 0$ (where $T_{\mu \nu}^i$ is still a valid stress-energy tensor, but it is not traceless).

We now come back the case of arbitrary $m \geq 0$ and $B$ with $f,g$ not supported in $\bar B$, and study the extra term needed in \eqref{Ineqbis}. 

\blem\label{dis}
Let $B$ be a bounded region and $f\in\S$. If $u\in \S$ and $u|_{\partial B} =  f|_{\partial B}$, then the function $\tilde f$ equal to f on $\bar B$ and to $u$ on $B^c$ belongs to $H_{m,+}$. 
\elem
\proof
$\tilde f$ clearly belongs to $L^2(\RR^d, d\bx)$.  Its partial derivatives $\partial_{x_i}\tilde f$ exist and are smooth out the measure zero set $\partial B$, and are bounded and rapidly decreasing at infinity, so also $\partial_{x_i}\tilde f \in L^2(\RR^d, d\bx)$.
Therefore the Fourier transforms satisfy 
\[\int_{\RR^d} |\,\widehat{\!\tilde f}(\bp)|^2 d\bp < \infty \, , \quad 
\int_{\RR^d} |\bp|^2|\,\widehat{\!\tilde f}(\bp)|^2 d\bp < \infty 
\]
 by the Plancherel theorem. This readily implies that $\tilde f\in H_+$ because
 \begin{align*}
\int_{\RR^d} \, \sqrt{|\bp|^2 + m^2}\, |\,\widehat{\!\tilde f}(\bp)|^2 d\bp &= \int_{|\bp| \leq 1} \, \sqrt{|\bp|^2 + m^2}\, |\,\widehat{\!\tilde f}(\bp)|^2 d\bp + \int_{|\bp| > 1} \, \sqrt{|\bp|^2 + m^2}\, |\,\widehat{\!\tilde f}(\bp)|^2 d\bp \\
&\leq  \sqrt{1+ m^2}\int_{|\bp| \leq 1} |\,\widehat{\!\tilde f}(\bp)|^2 d\bp + \int_{|\bp| > 1} \, (|\bp|^2 + m^2) |\,\widehat{\!\tilde f}(\bp)|^2 d\bp < \infty\, .
 \end{align*}
\eproof
\blem\label{dis2}
Let $\tilde f$ be as in Lemma \ref{dis}. Then $\tilde f$ is a measurable function, $\nabla \tilde f$ exists almost everywhere and
\ben\label{Su}
S(\tilde f , 0 | x_1 > 0) =\pi\int_{x_1 > 0} x_1 \big(|\nabla \tilde f|^2 + m^2 \tilde f^2\big)d\bx \, .
\een
\elem
\proof
By Lemma \ref{dis}, $\langle \tilde f, 0\rangle \in \H_m$, so $S(\tilde f , 0 | B)$ is well-defined. If $\tilde f \in\S$, then the lemma is a particular case of 
\cite[Thm. 5.4]{L19}, and the same proof works to show \eqref{Su} in our case. 
\eproof
\bprop\label{prop:Beknotloc}
Let $B\subset \RR^d$ be a region with half-width $R>0$ and $\Phi$ a Klein-Gordon wave packet with time-zero Cauchy data $f,g\in S(\RR^d)$. Then
\[
S(\Phi| B) \leq 2\pi R E(\Phi| B) + \Ga_\Phi\, ,
\]
where the constant $\Ga_\Phi$ depends only on the restriction $\Phi|_{\partial B} =f|_{\partial B}$ of $\Phi$ to the boundary $\partial B$ of $B$. 
\eprop
\proof
By applying a rotation and translation, we may assume that $B$ is contained in the strip $\{0\leq x_1 \leq 2R\}$. 

The inequality \eqref{i1}
\[
S(\Phi| B)  \leq S(\Phi | \{ x_1 > 0\}) = \pi\int_{x_1 > 0} x_1 \big(|\nabla f|^2 + m^2 f^2 + g^2\big)d\bx
\]
holds by the monotonicity of the entropy. 

By the locality of the entropy,  $S(f,g | B) $ remains unchanged if we replace $f,g$ with $u, v$, with $u,v\in S(\RR^d)$ and $u= f$, $v= g$ on $B$. 
Indeed $f - u$ and $g- v$ are supported in the complement $B^c$ of $B$ and $H_m(B^c)$ is contained in the symplectic complement $H_m(B)'$ of $H_m(B)$. 

Therefore,
\begin{align*}
S(\Phi| B) &\leq \pi\int_{x_1 > 0}x_1 \big(|\nabla u|^2 + m^2 u^2 + v^2\big)d\bx 
\\
&= \pi \int_B x_1{T_{00}}_\Phi(x)d\bx + \pi \int_{x\notin B,x_1 > 0}x_1 \big(|\nabla u|^2 + m^2 v^2 + v^2\big)d\bx 
\\
&= 2\pi R\int_B {T_{00}}_\Phi(x)d\bx + \pi \int_{x\notin B,x_1 > 0}x_1 \big(|\nabla u|^2 + m^2 u^2 + v^2\big)d\bx 
 \, .
\end{align*}
So
\[
S(\Phi| B) \leq 2\pi R \,E(\Phi |\! | B) + \pi \inf_{u,v} {\mathfrak I} (u,v)\, ,
\]
where
\[
\frak I (u,v) \equiv \int_{x\notin B,x_1 > 0}x_1 \big(|\nabla u|^2 + m^2 u^2 + v^2\big)d\bx
\]
and the infimum is taken over all $u,v\in S(\RR^d)$ such that
$u|_{\partial B} = f|_{\partial B}$, $v|_{\partial B} = g|_{\partial B}$. 

Now, ${\frak I} (u,v) = {\frak I} (u,0) + {\frak I} (0,v)$ and 
\[
\inf\big\{ {\frak I} (0,v): v\in \S, u|_{\partial B} = g|_{\partial B}\big\} = 0\, .
\]
By Lemmas \ref{dis} and  \ref{dis2}, we may further allow $u$ to be any vector in $H_{m,+}$. 
Therefore, we may choose
\[
\Ga_\Phi = \Ga_f := \inf\big\{ {\frak I} (u): u\in H_{m,+}, u|_{\partial B} = f|_{\partial B}\big\}
\]
where
\ben\label{Iu}
\frak I (u) \equiv \int_{x\notin B,x_1 > 0}x_1 (|\nabla u|^2 + m^2 u^2)d\bx\, ,
\een
that is, $\frak I (u) \equiv \frak I (u,0)$. Clearly, $\Ga_\Phi$ depends only on $f|_{\partial B}$. 
\eproof
\bcor
If $\Phi$ vanishes on the boundary $\partial B$ of $B$, then
\[
S(\Phi| B) \leq 2\pi R\, E(\Phi| B) \, ,
\]
namely $\Ga_\Phi = 0$. 
\ecor
\proof
If $f|_{\partial B} = 0$, we may choose $u$ to be identically zero on $B^c$, so $\frak I (u) = 0$. 
\eproof

\subsection{A variational problem}
Let $B$ be a bounded region of the half-space $x_1  >0$, with regular boundary (say $B$ is a ball), and $h$ a smooth function on the boundary $\partial B$ of $B$. 
We many consider $\frak I(u)$ \eqref{Iu} also as a functional on 
${\cal S}_h$,  the set of smooth functions $u$ on the interior of $\{x\notin B, x_1 >0\}$,  rapidly decreasing at infinity, continuous up to the boundary, such that $u|_{\partial B} = h$. Equivalently, we could consider $\frak I(u)$ \eqref{Iu} as a functional on the linear space $\{ u\in \S\, : u|_{\partial B} = h\}$.
Set
\[
\Ga_h \equiv \inf_{u\in {\cal S}_h} \frak I(u) \, .
\]
Of course, we may consider the functional $\frak I$ on different function classes. 
The problem is to estimate $\Ga_h$ in terms of $h$. 

Consider the Hilbert space $H_{x_1}$ (weighted Sobolev space) obtained as the completion of $\S$ with norm
\[
\|u\|_{x_1} \equiv \left(\int_{x\notin B, x_1 >0}x_1\big(m^2 u^2 + |\nabla u|^2\big)dx\right)^{1/2}\, .
\]

\bprop
Let $h$ be a smooth function on $\partial B$. 
There exists $u_h\in H_{x_1} $ that minimises $\frak I$, namely is $u_h |_{\partial B} = h$ and
\[
\frak I(u_h) = \Ga_h\, .
\]
The restriction of $u_h$ to the complement of $B$ in the half-space $\{x_1 >0\}$ is uniquely determined and satisfies
\ben\label{L}
\partial_{x_1} u_h = 
x_1(m^2 u_h- \nabla^2 u_h )   \, ,
\een
\eprop
\proof
The set $\S_h =\{u \in \S: u|_{\partial B} = h\}$, this should be a typo) is convex, so its closure has a vector $u_h\in \S_h$ of minimal norm.  Since $\frak I$ is strictly convex, $u_h$ is unique as element of $H^1$, i.e. the restriction of $u_h$ to $B^c\cap\{x_1 > 0\}$ is uniquely determined. 
Then, eq. \eqref{L} holds true. 
\eproof
Note that
\[
\Ga _h = \frak I(u_h) = \frac12 \int_{\partial B}x_1 \partial_{\bf n}(u_h^2)dS = \int_{\partial B}x_1 h\, \partial_{\bf n}(u_h)dS
\]
(normal derivative). We leave ot the problem to estimate this integral in terms of $h$. 

Note that $\Ga_h \equiv \Ga(h, m, B)$ satisfies the following properties. 
\bprop
We have:
\begin{itemize}
\item $\Ga(h, m, B)$ is convex both in $h$ and in $m$
\item $\Ga( \a h, m, B) = \a^2 \Ga(h, m, B)$, $\a\in \RR$
\item $\Ga(h, m, B)$ is non-decreasing in $m$, thus with a minimum at $m=0$
\item $\Ga(h_\l,  \l m, \l B) =\l^2 \Ga(h, m, B)$, where $\l > 0$ and $h_\l(x) = h(\l^{-1}x)$
\end{itemize}
\eprop

\subsection{Some remarks on the fermionic case}
A natural question is if the above analysis can be extended to the fermionic case. We discuss below the main obstructions to this end.

We start by fixing the notations we will use in this subsection. For simplicity, we restrict to $3+1$ dimensions, and, as usual, $g=[g_{\mu\nu}]_{\mu,\nu=0,\dots,3}$ will denote Minkowski metric on $\bR^4$, whose elements $x=(x^0,x^1,x^2,x^3) = (x^0, \bx)$ will have contravariant (upper) indices. Indices are lowered and raised by $g$ and $g^{-1}=[g^{\mu\nu}]$  respectively, and we employ the summation convention over repeated upper and lower indices, so that $x\cdot y = x_\mu y^\mu = x_0 y_0 - \bx \cdot \by$ is the Minkowski product. Given $\bp \in \bR^3$ and $m \geq 0$, we set $\mu_m(\bp) := \sqrt{|\bp|^2+m^2}$ as in Section \ref{SectWH}
and use the shorthand $p_+ := (\mu_m(\bp),\bp)$ for the 4-vector on the positive mass $m$ shell corresponding to $\bp$. The gamma matrices $\gamma^\mu$ satisfy
\[
\{ \gamma^\mu, \gamma^\nu \} = 2g^{\mu\nu},
\]
with $\{\cdot,\cdot\}$ the anticommutator, and we choose a representation in which the transpose of $\gamma^0$ satisfies $(\gamma^0)^t = \pm \gamma^0$. Moreover, we use the Dirac notation $\slashed{x} := x_\mu \gamma^\mu$ for $x \in \bR^4$. The charge conjugation matrix $C$ is such that $C^t = C^* = -C = C^{-1}$ and $C\gamma^\mu C^{-1} = -(\gamma^\mu)^t$. Finally, 3- and 4-dimensional Fourier transforms are defined respectively by 
\[
\hat g(\bp):= \int_{\bR^3} g(\bx)e^{-i\bp\cdot\bx}d\bx, \qquad \hat f(p) := \int_{\bR^4} f(x) e^{i p\cdot x} dx.
\]

Now, given a mass $m \geq 0$, we consider the real  linear space $\mathcal M \subset C^\infty(\bR^4,\bC^4)$ of solutions $\Psi$ of the Dirac equation
\begin{equation}\label{eq:Dirac}
(i \partial_\mu \gamma^\mu - m) \Psi = 0
\end{equation}
with Cauchy data $\Psi_0 := \Psi|_{x_0 =0} \in S(\bR^3, \bC^4)$ and satisfying the Majorana condition  $C(\gamma^0)^t \bar \Psi = \Psi$ (Majorana waves). $\cM$ becomes a complex space when endowed by the complex structure $\imath : \cM \to \cM$ defined by
\[
(\imath \Psi)_0\hat{\,}(\bp) := \frac i{\mu_m(\bp)}(m-p_j \gamma^j) \gamma^0\hat \Psi_0(\bp), \qquad \bp \in \bR^3.
\]
The universal covering of the Poincar\'e group naturally acts on $\mathcal M$, and defining the positive frequency part of $\Psi \in \cM$ as $\hat \Psi_+(p) := \theta(p_0) \hat \Psi(p)$, $p \in \bR^4$, the bilinear form
\[
( \Psi, \Phi) := \int_{x^0=0} d\bx \,\Psi_+(x)^* \Phi_+(x) = \frac 1{(2\pi)^3}\int_{\bR^3} \frac{d\bp}{2\mu_m(\bp)}\hat\Psi_0(\bp)^* (\slashed{p}_++m)\gamma^0 \hat\Phi_0(\bp), 
\]
is actually sesquilinear with respect to the complex structure $\imath$, and defines a Poincar\'e invariant scalar product on $\cM$. We denote by $\cK_m$ the complex Hilbert space obtained by completing $\cM$ with respect to this scalar product. For a ball $B \subset \bR^3$, the closed real subspace of $\cK_m$
\[
K_m(B) := \{ \Psi \in \cM \,:\, \supp \Psi_0 \subset B\}^-
\]
is standard.

The Majorana field is the (spinor valued) Wightman field $\psi$ on the fermionic Fock space $\Gamma_-(\cK_m)$ defined as follows: for $f \in S(\bR^4,\bC^4)$ such that $C (\gamma^0)^t \bar f = f$ define $\Psi_f \in \cM$ by
\[
\Psi_f(x) = i \int_{\bR^4} dy\, (i \gamma^\mu \partial_\mu + m)D(y-x) \gamma^0\overline{f(y)}, \qquad x \in \bR^4,
\]
with $D$ the causal propagator of the mass $m$ Klein-Gordon field; 
then the Majorana field smeared with $f$ is
\[
\psi(f) := a(\Psi_f) + a(\Psi_f)^*,
\]
with $a(\Psi)$, $\Psi \in \cK_m$, the generators of the CAR algebra on $\Gamma_-(\cK_m)$; for arbitrary $f \in S(\bR^4,\bC^4)$, the smeared field $\psi(f)$ is obtained by complex linearity. The local von Neumann algebra of the Majorana field associated to a double cone $O$ with basis a ball $B$ in the $x^0 = 0$ plane is then
\[
\cF(O) := \{ \psi(f)+\psi(f)^* \,: \, \supp f \subset O\} '' = \{ a(\Psi) + a(\Psi)^*\,:\, \Psi \in K_m(B)\}'',
\]
the fermionic second quantization of $K_m(B)$. Algebras for more general double cones are obtained by Poincar\'e covariance, and for more general regions by additivity.

The first problem that one has to face in trying to obtain some kind of Bekenstein bound for Majorana waves is the absence of a meaningul notion of entropy for them. In particular, it is not clear if there is a suitably large class of states of the Majorana field whose relative entropy with respect to the vacuum can be expressed directly in terms of Majorana waves, as in the case of coherent states of the Klein-Gordon field (note that the fermionic coherent states usually employed in the physics literature involve auxiliary ``Grassmann variables'', and are therefore not represented by vectors in $\Gamma_-(\cK_m)$).

One can then ask if a Bekenstein bound can be obtained for the relative entropy of some class of states at the second quantization level.  To this end, it is natural to consider first the class of one particle states. Indeed, due to the CARs, the square of  $\psi(f) + \psi(f)^*$ is a multiple of the identity, and then states of the form $e^{i (\psi(f)+\psi(f)^*)}\Omega$ are actually linear combinations of the vacuum and a one-particle state.

Thanks to the Bisognano-Wichmann theorem, it is then not difficult to show, similarly to~\cite{PM25} (where the double cone, $m=0$ case is considered), that the relative entropy between a one particle state $\omega_\Psi = ( \Psi, (\cdot) \Psi)$, $\Psi \in \cM$ with $\supp \Psi_0 \subset \{ x^1 > a \}$, $a \in \bR$, and the Fock vacuum $\omega$ on the algebra $\cF(W_a)$ associated to the wedge $W_a := \{ x^1 > a\}''$ is given by
\begin{equation}\label{eq:Sfermionic}
S( \omega_\Psi |\!| \omega)_{\cF(W_a)} = \pi \int_{x^0 = 0} d\bx\, (x^1-a) (\Psi, T_{00}(x)\Psi),
\end{equation}
with
\[
T_{00}(x) = \frac i 2 [: \psi^*(x) \partial_0 \psi(x) : - :\partial_0 \psi^*(x) \psi(x):], \qquad x \in \bR^4,
\]
the quadratic form expressing the quantized energy density of the Majorana field. The above expression for $S(\omega_\Psi |\!| \omega)_{\cF(W_a)}$ is of course similar to~\eqref{SPhi}. One obtains then the following analogue of Prop.~\ref{bek1}.

\bprop
If $\Psi \in \cM$ is such that $\supp \Psi_0 \subset \bar B$, $B \subset \bR^3$ a ball of radius $R > 0$, and $O$ is a  double cone with basis $B$, then
\[
S(\omega_\Psi |\!\ \omega)_{\cF(O)} \leq 2\pi R E(\Psi | B)
\]
with $E(\Psi |B) := \int_B (\Psi, T_{00}(0,\bx)\Psi) d\bx$.
\eprop

\proof
Arguing as in the proof of Prop.~\ref{bek1}, it is sufficient to show that the integrand in~\eqref{eq:Sfermionic} vanishes outside of $B$. Indeed, one finds
\[\begin{split}
(\Psi&, T_{00}(0,\bx)\Psi) \\
&= \frac1{4(2\pi)^6} \int_{\bR^6} d\bp d\bq\left(\frac1{\mu_m(\bp)}+\frac 1{\mu_m(\bq)}\right)\hat\Psi_0(\bp)^*(\slashed{p}_+ + m)\gamma^0 (\slashed{q}_+ + m)\gamma^0 \hat\Psi_0(\bq) e^{i(\bq-\bp)\bx}\\
&=\frac1{4(2\pi)^6} \int_{\bR^6} d\bp d\bq\hat\Psi_0(\bp)^*\left[\left(\frac1{\mu_m(\bp)}+\frac 1{\mu_m(\bq)}\right)p_j q_k \gamma^j \gamma^k+ (q_k+p_k)\gamma^0 \gamma^k\right]\hat\Psi_0(\bq) e^{i(\bq-\bp)\bx}\\
&-\frac1{4(2\pi)^6} \int_{\bR^6} d\bp d\bq\hat\Psi_0(\bp)^*\left[\frac{p_k}{\mu_m(\bp)}\mu_m(\bq)+\mu_m(\bp)\frac{q_k}{\mu_m(\bq)}\right]\gamma^k\gamma^0 \hat\Psi_0(\bq) e^{i(\bq-\bp)\bx},
\end{split}\]
where the second equality follows from the fact that the first integral is manifestly real. Now, while the first integral in the right hand side can clearly be expressed as a sum of terms of the form
\[
(G*P_1 \Psi_0)(\bx)^* P_2 \Psi_0(\bx) \quad \text{or}\quad P_1'\Psi_0(\bx)^*(G'*P_2' \Psi_0)(\bx)
\]
with $G, G'$ tempered distributions and $P_j, P_j'$, $j=1,2$ differential operators, and therefore has support contained in that of $\Psi_0$, the second integral equals
\[
2\Re\left[-i (\mu_m \Psi_0)(\bx)^*\gamma^0 \gamma^k \partial_k(\mu_m^{-1} \Psi_0)(\bx)\right]
\]
which vanishes thanks to the Majorana condition. 
\eproof

It would then be interesting to try to also extend to the fermionic case the bound of Prop.~\ref{prop:Beknotloc} for non localized one particle states. However, while of course~\eqref{eq:Sfermionic} is expected to hold also in this case, its proof crucially depends on the fact that $\Psi$ can be obtained by acting on the vacuum with a unitary operator localized in $W_a$, which allows to express the relative modular operator between $\Psi$ and $\Omega$ in terms of the modular operator  for $\Omega$. Unfortunately, already for one particle states which are sums of a state with Cauchy data supported in the basis $\{x^1 > a\}$ of $W_a$ and one with Cauchy data supported in the complement, the problem of computing the relative modular operator becomes much more difficult.

\section{The entropy balance and ant formulas for wave packets}\label{eba}
 Let $\A$ be a Poincaré covariant, wedge dual net of von Neumann algebras on a Hilbert space. We restrict attention to $1+1$ dimensional Minkowski spacetime with coordinates $x = \langle x_0, x_1\rangle \in \RR^2$, and we denote by $\A{(x)} = \A(W{(x)})$  the von Neumann algebra associated with the right wedge $W(x)$ with vertex $x$, and by $U$ the 
translation unitary group (for simplicity, we consider the Minkowski space time in dimension $1+1$, the general case follows at once from this). Given a vector $\Phi\in\H$, we consider the relative entropy between the vector states $\Phi$ and $\Om$ on $\A(x)$ and $\A(x)'$ that we denote by
\ben\label{SSb}
S(x) = S(\Phi |\!| \Om)_{\A(x)}\, ,\qquad {\bar S}(x) = S(\Phi |\!| \Om)_{\A(x)'}\, .
\een
The {\it entropy balance} formula for light-like translations states that 
\ben\label{dS0}
S(a,a) - S(b,b)   =   {\bar S}(a,a) - {\bar S}(b,b) + 2\pi(b-a) (\Phi, P_+\Phi) \, ,
\een
$a,b\in\RR$, where $P_+$ is the generator of the right null-translation one-parameter unitary group, and was proved in \cite{CF18}. 

The version of the above formula for spatial translations
\ben\label{dS1}
S(0,a) - S(0,b)   =   {\bar S}(0,a) - {\bar S}(0,b) + 2\pi(b-a) (\Phi, P\Phi) \, ,
\een
with $P$ is the Hamiltonian, namely the generator of the time-translation unitary group, was obtained in \cite{HL}.

The {\it ant formula} states that,
if $\partial_a S(a,a)$ exists for some $a \in \RR$,  then 
\ben\label{antleft}
-\partial_a S(a,a) = 2\pi \inf_{u' \in \A(a,a)'} (u'\Phi, P_+ u'\Phi)\, ,
\een
 where the infimum is taken over isometries $u' \in \A(a,a)'$. The ant formula was considered in \cite{W17} and proved in \cite{CF18}; a different proof was given in \cite{HL}. 

Both formulas  \eqref{dS1} and \eqref{antleft} require the involved quantities to be well-defined. The ant formula has so far been proved for a dense subspace of vectors of the Hilbert space (possibly not the same in \cite{CF18} and \cite{HL}).

We are going to discuss the version of the above formulas in the wave packet context. 
Let $H$ be the local, Poincaré covariant net of standard subspaces on the wave Hilbert space $\H$ as in Section \ref{SB}. Set $H(x) \equiv H(W(x))$, with associated entropy operator $\E(x) = \E_{H(x)}$.  Therefore, from \eqref{EE'} we get
\ben
\E(x)  - \E(y)  = \bar\E(x)  - \bar\E(y) - \big(\log\Delta_{H(x)} - \log\Delta_{H(y)}\big)\, ;
\een
here, $x = \langle x_0,x_1\rangle$ and $y = \langle y_0,y_1\rangle$ are arbitrary points of the Minkowski plane $\RR^2$. So, by the Bisognano-Wichmann property, we have
\ben
\E(x)  - \E(y)  = \bar\E(x)  - \bar\E(y) -2\pi \big(K(y) - K(x)\big)\, ,
\een
where $ K(x)$ is the generator of the one-parameter unitary group of the boosts associated with the wedge $W(x)$. Now,
\[
K(y) - K(x) = -(y_1 - x_1) P - (y_0 - x_0) P_1
\]
with $P_1 = P_+  - P$ the generator of $x_1$-translation one-parameter unitary group, therefore
\ben\label{gef}
\E(x)  - \E(y)  = \bar\E(x)  - \bar\E(y) + 2\pi (y_1 - x_1) P + 2\pi (y_0 - x_0) P_1\, .
\een
We then have the following general entropy balance formula for wave packets where, given $\Phi \in \H$, by a slight abuse of notation we denote also by $S(x) := S(\Phi | H(x))$ the entropy of $\Phi$ with respect to $H(x)$.
\bprop\label{BE}
Let $\Phi\in\H$ be a vector in the domain of all operators in formula \eqref{gef}. Then
\ben\label{gBEE}
S(x)  - S(y)  = \bar S(x)  - \bar S(y) + 2\pi (y_1 - x_1) (\Phi, P\Phi) + 2\pi (y_0 - x_0) (\Phi, P_1 \Phi)\, ,
\een
with the notations in \eqref{SSb}. 
\eprop
\proof
It is enough to take the expectation values on $\Phi$ on both sides of the identity \eqref{gBEE}. 
\eproof
We now consider the ant formula in our context. Set for simplicity $S(a) = S(a,a)$, $a\in\RR$.  Proposition \ref{BE} gives
\[
S(a) - S(b) = \bar S(a) - \bar S(b) + 2\pi(b - a)(\Phi, P_+ \Phi) 
\]
so, if $b > a$, we have
\ben\label{ri}
-\frac{S(a) - S(b)}{a - b} \leq 2\pi(\Phi, P_+ \Phi) 
\een
by the monotonicity of the relative entropy. As $S$ is monotone, it is almost everywhere differentiable. By \eqref{ri}, if $S(a)$ is differentiable in $a$ and $\Phi$ is in the domain of $P_+$ and of $\E(b)$ for $b$ in a right neighbourhood of $a$, we have 
$- \partial_a S(a) \leq 2\pi(\Phi, P_+ \Phi)$, hence
\ben\label{ind}
-\partial_a S(a) \leq 2\pi \inf_{\Psi \sim_a\Phi}  (\Psi, P_+ \Psi) \, ,
\een
where $\Psi \sim_a \Phi$ means that $\Psi$ is in the domain of $P_+$ and of $\E(b)$ for $b$ in a right neighbourhood of $a$ and $\Psi - \Phi\in H'(a,a)$ , because the entropy $S(a)$ depends only on the class of $\Phi$ in the above sense. 

The following ant formula will show that the inequality \eqref{ind} is actually an equality. Note that this equality does not follow from the ant formula for the second quantized net von Neumann algebras; the latter is indeed a weaker inequality because there appear states that are not coherent. 
\bprop
\label{prop5.2}
Let $H$ be a local, Poincaré covariant net of standard subspaces on wedges on the Hilbert space $\H$; we assume the Bisognano-Wichmann property. With the above notations and with $\Phi$ a vector in the domain of $P_+$ and of $\E(b)$ for $b$ in a right neighbourhood of $a$,
\ben\label{ant2}
-\partial_a S(a) = 2\pi\inf_{\Psi \sim_a\Phi}  (\Psi, P_+ \Psi) \, ,
\een
at every point $a$ such that $\partial_a S(a)$ exists. 

As a consequence, $\partial_a S(a)$ is increasing, so $S$ is a convex function. 
\eprop
\proof
We consider the net $a\in \RR \mapsto H(a) \equiv H(a,a)$ of standard subspaces on half-lines of $\RR$, and prove the statements for this one-dimensional net $H(a)$. Now $H(a)$ is, up to multiplicity, the conformal net associated with the $U(1)$-current algebra, see \cite[Sect. 2.6.1]{L08}, and we thus assume $H(a)$ is this net.
As discussed in App.~\ref{app:U1}, in this case Hilbert space vectors  $\Phi = \Phi_f$ can be described by (equivalence classes) of real distributions $f$ on $\bR$, and if $\Phi_f$ satisfies the assumptions then $\partial_a S(a)$ exists if and only if $f' \in L^2((a,+\infty), dx)$, and
\[
-\partial_a S(a) =  \pi\int_a^\infty f'(t)^2 dt.
\]
This easily implies
\[
- \partial_a S(a) =  \pi \inf_{h \sim f}\int_{-\infty}^{+\infty} h'(t)^2 dt = 2\pi \inf_{h \sim f} (\Phi_h, P_+ \Phi_h)
= 2\pi \inf_{\Psi\sim_a\Phi_f} (\Psi, P_+ \Psi)\, ,
\]
where the infimum is taken over all real $h$ in the Hilbert space such that $h'\in L^2(\RR)$, and $h \sim f$ means that $h = f$ on $[a,+\infty)$.  The third equality follows from the second one and the inequality \eqref{ind}. 
 
The last part of the statement now follows at once, a direct proof is contained in \cite{CLRR22}. 
\eproof

By a generalization of the above proof, one may obtain a 2-dimensional version of proposition \ref{prop5.2}, as follows. 
Alternatively to the net $a \in \RR \mapsto H^+(a) \equiv H(x^0 +a, x^1+a)$ of standard subspaces, similar up to a translation 
by $x = \langle x^0, x^1 \rangle \in {\mathbb R}^2$ to that considered in the proof, we may consider 
$a \in \RR \mapsto H^-(a) \equiv H(x^0-a, x^1+a)$. Then for the general 2-vector $x \in \RR^2$, we define the equivalence relation
$\Psi \sim_x \Phi$ to mean that $\Psi - \Phi\in H(x)'$, and we define the entropy $S(x)$ associated 
with the standard subspace $H(x)$ and vector $\Phi$ as above. Then, 
applying proposition \ref{prop5.2} to the nets $a \in {\mathbb R} \mapsto H^\pm(a)$ in 
turn and adding up the resulting variational principles with the weights $s^+, -s^-\ge 0$ yields 
\ben\label{ant3}
-s^\mu \partial_\mu S(x) = 2\pi\inf \{ s^+(\Psi_+, P_+ \Psi_+) - s^-(\Psi_-, P_- \Psi_-) \ : \ \Psi_\pm \sim_{x} \Phi \}
 \, ,
\een
where $P_- = P-P_1 \ge 0$, $s = \langle s^0, s^1 \rangle$ is a spacelike vector pointing to the right such that $s^\pm = (s^0 \pm s^1)/2$ are its null-components. By considering 
the 2-dimensional variational principle \eqref{ant2} instead of $x$ for a $y \in \RR^2$ to the spacelike right of $x$ (so that one is varying over a smaller set of vectors on the right side of \eqref{ant3}), one obtains a convexity property of $x \mapsto S(x)$:
\ben
\label{Sconv}
S(x+r+s)+S(x) \ge S(x+r)+S(x+s)
\een
for all spacelike vectors $r,s$ pointing to the right. This property can be seen as a generalization of \eqref{qdec1}. It is related to QDEC which we will discuss in \cite{HLM}.

\appendix

\section{Elementary formulas}
Let $f,g,u$ be smooth functions and $\bf V$ a smooth vector field on $\bar B$. By the divergence theorem, we have
\[
\int_B f\, {\rm div}({\bf V})d\bx = \int_{\partial B} f\, {\bf V}\cdot{\bf n}\,dS - \int_B \nabla f \cdot {\bf V} d\bx
\]
where $\bf n$ is the outward unit normal vector to the boundary, integrated with respect to its standard Riemannian volume form 
$dS$ on $\partial B$. 

Taking ${\bf V} = u\nabla g$, we have
\[ 
{\rm div}({\bf V}) = {\rm div}(u\nabla g) = \sum_i \partial_i(u \partial_i g) = \nabla u \nabla g + u \nabla^2 g
\]
\[
\int_B uf\, \nabla^2 g\, d\bx + \int_B f \nabla u \nabla g\, d\bx
= \int_{\partial B}u f\, \partial_{\bf n}g\,dS - \int_B u  \nabla  f  \nabla g \, d\bx \, ;
\]
setting $u(\bx) = 1 + r^2$, we hen have
\[
\int_B (1 + r^2)  \nabla  f  \nabla g \, d\bx 
= -\int_B (1 + r^2)f\, \nabla^2 g\, d\bx  - 2\int_B f\, r \partial_r g\, d\bx
+ \int_{\partial B}(1 + r^2) f\, \partial_{\bf n}g\,dS  \, ,
\]
that is
\[
\int_B (1 + r^2)  \nabla  f  \nabla g \, d\bx 
= (f, T g)
+ \int_{\partial B}(1 + r^2) f\, \partial_{\bf n}g\,dS  \, ,
\]
with $T = -(1 + r^2)\nabla^2 - r\partial_r$. 

We are interested in the case$f=g$, $u(\bx) = x_1$ too; we have:
\[
\int_B x_1 f\, \nabla^2 f\, d\bx + \int_B f  \partial_{x_1} f\, d\bx
= \int_{\partial B}x_1 f\, \partial_{\bf n}f\,dS -  \int_B x_1 | \nabla  f|^2   \, d\bx \, ;
\]
\medskip

\section{Finite entropy vectors for  the U(1) current}\label{app:U1}
 In this appendix, we compute the vectors of the one particle space of the U(1) current on the real line whose entropy with respect to the standard subspace of a bounded interval or of a half-line is finite. 
 
 In order to fix notations, we start by reviewing the construction of the net of standard subspaces of the U(1) current. Our convention for the Fourier transform  and anti-transform on $\bR$ is
 \[
 \hat \varphi(p) =\frac1{\sqrt{2\pi}} \int_\bR \varphi(x)e^{ipx}dx, \qquad \check \varphi(x) = \frac1{\sqrt{2\pi}} \int_\bR \varphi(p) e^{-ipx} dp.
 \] 
 Moreover all function spaces considered will consist of real valued functions.

 We denote by $S_1(\bR)$ and $\hat S_1(\bR)$ the subspaces of $S(\bR)$ of the (real) functions $\varphi \in S(\bR)$ such that
 \[
 \int_\bR \varphi(x)\,dx = 0 \qquad\text{and}\qquad   \varphi(0) = 0,
 \] 
 respectively, with the induced locally convex topology. Of course $\varphi \in S_1(\bR)$ if and only if $\hat \varphi \in \hat S_1(\bR)$. 
  Therefore, given $f \in S'_1(\bR)$, its Fourier transform is the $\hat f \in \hat S'_1(\bR)$ and its derivative is the $f' \in S'(\bR)$ defined by duality:
 \[
 \langle \hat f,  \varphi\rangle = \langle f, \hat \varphi \rangle, \quad \varphi \in \hat S_1(\bR); \qquad \langle f',\varphi\rangle =-\langle f,\varphi'\rangle, \quad \varphi \in  S(\bR).
 \]
We consider then the real linear space $\fL$ of distributions $f \in S_1'(\bR)$ whose Fourier transform $\hat f$ is a measurable function on $\bR$ such that
\[
\| \hat f \|^2:= \int_0^{+\infty} p |\hat f(p)|^2 dp < +\infty.
\] 

\begin{proposition}
The space $\fL$ equipped with the real scalar product
\[
\langle f,g\rangle:= \Re\left[ \int_0^{+\infty} p \hat f(-p) \hat g(p)\,dp\right], \qquad f,g \in \fL,
\] 
is a real Hilbert space.
\end{proposition}

\proof
If $\{ f_n\} \subset \fL$ is a Cauchy sequence, $\{ \hat f_n\} \subset L^2(\bR_+, p dp)$ is also Cauchy (with respect to the $L^2$ norm). Let $\hat f \in L^2(\bR_+,p dp)$ be its limit. Now if $\varphi \in S_1(\bR)$ then 
\[
\langle \hat f, \check \varphi \rangle := \int_\bR \hat f(p) \check \varphi(p)dp = \frac1{2i} \Im\left[ \int_0^{+\infty} p \overline{\hat f(p)}\frac{\check \varphi(-p)}{p}\right]
\]
is well defined, since 
$\check \varphi \in \hat S_1(\bR)$ satisfies $\check \varphi(p) \sim p$ for $p \to 0$, and therefore $p \mapsto \check \varphi(-p)/p$ belongs to $L^2(\bR_+,p dp)$. As a consequence $\hat f \in \hat S_1'(\bR)$, and $f \in S_1'(\bR)$ is the limit of $\{f_n\}$ in $\fL$.
\eproof

We also consider the complex structure $\imath : \fL \to \fL$ defined by $(\imath f)\hat{\,}(p) = i \operatorname{sign}(p) \hat f(p)$, $f \in \fL$. Thus $\fL$ becomes a complex Hilbert space with the complex scalar product
\[
(f,g) := \langle f,g\rangle - i\langle f,\imath g\rangle = \int_0^{+\infty} p \hat f(-p) \hat g(p)\,dp, \qquad f, g \in \fL,
\]
and
\[
\beta(f,g) := \Im (f,g) = -\langle f,\imath g\rangle = \frac{1}2 \int_\bR \hat f(-p) \widehat {g'}(p)\,dp, \qquad f,g \in \fL,
\]
is a symplectic form on $\fL$.

\begin{proposition}\label{prop:density}
$C^\infty_c(\bR)$ is dense in $\fL$.
\end{proposition}

\proof
We start by observing that for $f \in \fL$ one has
\[
\| f - (\chi_{[\varepsilon,+\infty)} \hat f)\check{\,}\|^2 = \int_0^{\varepsilon} p |\hat f(p)|^2dp \to 0
\]
as $\varepsilon \to 0$, and it is then sufficient to prove that for any $\varepsilon > 0$, $C^\infty_c(\bR)$ is dense in the subspace of elements $f \in \fL$ whose Fourier transform vanishes in $(-\varepsilon, \varepsilon)$. Let $f$ be such an element of $\fL$. Then, clearly, $\hat f \in L^2(\bR)$ and therefore $f \in L^2(\bR)$. Now, given a non negative $\rho \in C^\infty_c(\bR)$ with support in $[-1,1]$ and such that $\int_\bR \rho = 1$, a sequence of mollifiers is obtained by $\rho_n(x) := n \rho(nx)$, $x \in \bR$, $n \in \bN$. Then, for all $n \in \bN$, $\rho_n * f$ is in $C^\infty(\bR) \cap L^2(\bR) \subset \fL$ with all its derivatives $(\rho_n * f)^{(h)} = \rho_n^{(h)}*f$ (by the Young inequality), and
\[
\|\rho_n * f - f\|^2 = \int_0^{+\infty} p |\sqrt{2\pi}\hat \rho(p/n)-1|^2 |\hat f(p)|^2 \to 0
\]
as $n \to +\infty$, by dominated convergence. It is then sufficient to show that if $f$ belongs to $C^\infty(\bR) \cap L^2(\bR)$ with all its derivatives it can be approximated in $\fL$ with $C^\infty_c(\bR)$ functions. To this end, take a non negative $\chi \in C^\infty_c(\bR)$ such that $\supp \chi \subset [-2,2]$ and $\chi = 1$ on $[-1,1]$, and define a sequence of cut-off functions by $\chi_n(x) := \chi(x/n)$, $n \in \bN$, $x \in \bR$. Then $g_n := \chi_n f-f \in L^2(\bR) \cap C^\infty(\bR)$ and
\begin{equation}\label{eq:normL2bound}
\| g_n\|^2 = -\beta(\imath g_n, g_n) = -\frac{1}2 \int_\bR (\imath g_n)(x) g_n'(x)\,dx \leq \frac12 \| g_n\|_{L^2(\bR)} \| g'_n\|_{L^2(\bR)},
\end{equation}
so it is sufficient to show that $g_n, g'_n \to 0$ in $L^2(\bR)$. It is clear that $g_n \to 0$. Moreover,
\[
\| \chi'_n f \|_{L^2(\bR)}^2 = \frac1{n^{2}} \int_{\bR} | \chi'(x/n) f(x)|^2 dx \leq \frac{\|\chi'\|_\infty^2}{n^{2}} \|f\|^2_{L^2(\bR)}\to 0,
\]
while
\[
\| (\chi_n-1) f'\|^2_{L^2(\bR)} \leq \int_{|x|> 2n} |f'(x)|^2dx \to 0.
\]
Thus $g'_n \to 0$ in $L^2(\bR)$ as $n \to +\infty$, and the proof is complete.
\eproof

Since the only constant $f \in \fL$ is everywhere vanishing, we can consider the (algebraic) direct sum $X:= \fL+ \bR$, and its quotient $\cH := X/\bR$  coincides with the set of equivalence classes of elements $f \in \fL$, which we will denote by $[f]$, or just by $f$ if there is no danger of confusion. Therefore, the complex structure $\imath$, the complex scalar product $(\cdot,\cdot)$ and the symplectic structure $\beta$ all lift to $\cH$ by
\[
\imath [f] := [\imath f], \quad ( [f], [g] ) := (f,g), \quad \beta([f], [g]) := \beta(f,g), \quad f, g \in \fL,
\]
thus making $\cH$ a complex Hilbert space.

The reason for passing from $\fL$ to $\cH$ is that on $\cH$ there is a natural unitary action of $\PSL2$, defined as follows. We identify the circle $\bS^1$ (tought as a subset of $\bC$) with the one point compactification $\bar \bR$ of $\bR$ via the stereographic projection, which sends $-1 \in \bS^1$ to $\infty \in \bar\bR$. As a consequence, we identify functions on $\bS^1$ and functions on $\bar \bR$ (and their restrictions to $\bR$). 

\begin{proposition}
If $f \in C^\infty(\bS^1)$, then $f - f(\infty) \in \fL$. 
\end{proposition}
\proof
Indeed, $g(x) := f(x) - f(\infty) \sim \frac 1{x}$ for $|x| \to +\infty$, which entails $g \in L^2(\bR)$. Moreover, since $g$ is smooth, $\hat g$ is rapidly decreasing, and then $\hat g \in L^2(\bR_+,p\, dp)$.
\eproof

Thus, by Prop.\ \ref{prop:density}, $\cH$ is the completion of $C^\infty(\bS^1)$ with respect to the norm $\|\cdot\|$, and then it coincides with the one particle space of the U(1) current~\cite{GLW98}. The unitary, irreducible, positive energy representation $U$ of the M\"obius group $\PSL2$ with lowest weight $1$ acts then on $\cH$, and is fixed by
\[
(U(g)f)(x) =  f(g^{-1}x), \qquad f \in C^\infty(\bS^1), \,x \in \bR,
\]
where $g = \left[\begin{matrix} a&b\\c&d\end{matrix}\right] \in \SL2$ acts on $\bR$ as $gx = \frac{ax+b}{cx+d}$, $x \in \bR$. 

\begin{proposition}\label{prop:CScomplex}
$C^\infty(\bS^1)$ is a $U$-invariant, complex subspace of $\cH$.
\end{proposition} 
\proof
$U$-invariance is obvious. Next, if $f_k(z) : = z^k$, $z \in \bS^1$, $k \in \bZ$, it is easy to see, computing the Fourier transform (as functions of $x \in \bR$) by contour integration, that $(\iota (f_k-f_k(\infty)))\hat{\,}(p) = i \operatorname{sign}(p) (f_k-f_k(\infty))\hat{\,}(p)=i \operatorname{sign}(k) (f_k-f_k(\infty))\hat{\,}(p)$, $k \in \bZ$, $p \in \bR$. Moreover, $f \in C^\infty(\bS^1)$ if and only if it can be expanded as $f = \sum_{k \in \bZ} c_k f_k$ with rapidly decreasing $c_k = \overline{c_{-k}}\in \bC$, and, since we are quotienting out constants in $\cH$, we can assume $c_0= 0$. Then, by dominated convergence, the series $f-f(\infty) = \sum_{k \in \bZ} c_k (f_k - f_k(\infty))$ converges in $L^2(\bR,dx)$, and the same applies to $(f-f(\infty))\hat{\,} = \sum_{k \in \bZ} c_k (f_k-f_k(\infty))\hat{\,}$. In turn, this implies that $(f-f(\infty))\hat{\,}(p) = \sum_{k \in \bZ} c_k (f_k-f_k(\infty))\hat{\,}(p)$ for almost every $p \in \bR$, and thus, as elements in $\cH$, $\iota f = \sum_{k \in \bZ} i \operatorname{sign}(k) c_k f_k \in C^\infty(\bS^1)$.
\eproof

Let $\cI$ be the set of proper, open, non-dense intervals of $\bS^1 \equiv \bar \bR$. To each $I \in \cI$ we can associate the closed real subspace of $\cH$
\[
H(I) := \{ f \in \cH \, :\, f \in C^\infty(\bS^1), \,\supp f \subset I \}^- 
\]
(norm closure), thus obtaining  a M\"obius covariant net of closed real subspaces $I \mapsto H(I)$. As a consequence, for all $I \in \cI$, $H(I)$ is standard in $\cH$ and satisfies the Bisognano-Wichmann property
\[
\Delta_{H(I)}^{it} = U(	\delta_I(-2\pi t)), \qquad t \in \bR,
\]
with  $\delta_I : \bR \to \SL2$ the one parameter group of ``dilations'' of the interval $I$, namely $\delta_I(s) = g_I^{-1} \delta_{\bR_+}(s) g_I$, $s \in \bR$, where $g_I \in \SL2$ is such that $g_I I= \bR_+$, and $\delta_{\bR_+}(s) x = e^s x$, $x \in \bR$. In particular, for $B := (-1,1) \subset \bR$ one gets
\[
\delta_B(s)x =  \frac{\cosh(s/2) x + \sinh(s/2)}{\sinh(s/2)x+\cosh(s/2)} = \frac{1+x-e^{-s}(1-x)}{1+x+e^{-s}(1-x)}, \qquad x,s \in \bR.
\]

\begin{proposition}
For all $I \in \cI$, $C^\infty(\bS^1)$ is a core for the modular Hamiltonian $\log\Delta_{H(I)}$. Moreover, for $f \in C^\infty(\bS^1)$,
\begin{equation}\label{eq:logDelta}
(\iota \log\Delta_{H(\bR_+)}f)(x) = 2\pi x f'(x), \qquad (\iota \log \Delta_{H(B)}f)(x)= 2\pi (1-x^2) f'(x), \qquad x \in \bR.
\end{equation}
\end{proposition}

\proof
By Prop.~\ref{prop:CScomplex}, it is sufficient to show that $C^\infty(\bS^1)$ in contained in the domain of $\log\Delta_{H(I)}$, and by M\"obius covariance, it is sufficient to show it for $I = \bR_+$. To this end, observe first that since $f \in C^\infty(\bS^1)$ entails $f'(x) \sim \frac c{x^2}$ for $|x| \to +\infty$, the function $h(x) := x f'(x)$, $x \in \bR$, belongs to $C^\infty(\bR) \cap L^2(\bR)$, and therefore belongs to $\fL$. We define then
\[
g_t(x) := \frac1 t[\Delta_{H(\bR_+)}^{it}f-f](x)-2\pi h(x) = \frac{2\pi}t \int_0^t ds\,[h(x)-e^{2\pi s} h(e^{2\pi s}x)], \qquad x \in \bR,
\]
and observe that, thanks to~\eqref{eq:normL2bound},
\[
\left\| \frac1{t}[\Delta_{H(\bR_+)}f-f]-2\pi h\right\|^2 = \| g_t\|^2 \leq \frac12 \|g_t\|_{L^2(\bR)} \|g'_t\|_{L^2(\bR)}.
\]
Now clearly $g_t \to 0$ pointwise in $\bR$ as $t \to 0$, and from the fact that $h(x) \sim \frac c x$ for $|x| \to +\infty$ and that $h$ is bounded on $\bR$, it is easy to see that $g_t$ can be bounded by an $L^2(\bR)$ function uniformly for $t$ in a neighborhood of zero, thus $\| g_t \|_{L^2(\bR)} \to 0$ by dominated convergence. Since $h'(x) = 2f'(x) +x f''(x) \sim \frac{c'}{x^2}$ as $|x| \to +\infty$, the same argument applies to show that also $\|g'_t\|_{L^2(\bR)} \to 0$, thus proving that $f \in D(\log \Delta_{H(\bR_+)})$ and that the first formula in~\eqref{eq:logDelta} holds. The second formula in~\eqref{eq:logDelta} is obtained analogously, taking into account that $x \mapsto (1-x^2) f'(x)$ belongs to $\fL+\bR$.
\eproof

We also put on record the following general fact.

\begin{proposition}
Let $H$ be a standard subspace of the complex Hilbert space $\cH$. Then any core for $\log \Delta_H$ is also a core for the entropy operator $\cE_H$.
\end{proposition}

\proof
We recall~\cite[Thm.\ 2.2]{CLR20} that $\cE_H$ is the closure of
\[
- (1-\Delta_H)^{-1}\log \Delta_H + J (1-\Delta_H)^{-1} \Delta_H^{1/2} \log \Delta_H,
\]
and since the functions $\l \mapsto \frac{\log \l}{1-\l}$, $\l \mapsto \frac{\sqrt{\l} \log \l}{1-\l}$ are continuous on $\bR_+$ and bounded by $c |\log \l|$ for $\l \to 0^+$ and $\l \to +\infty$, the statement follows by spectral calculus.
\eproof

\begin{corollary}\label{cor:coreE}
For all $I \in \cI$, $C^\infty(\bS^1)$ is a core for the entropy operator $\cE_{H(I)}$.
\end{corollary}

We now consider the real quadratic forms $q_{\bR_+} : D(q_{\bR_+}) \to \bR$ and $q_B : D(q_B) \to \bR$ on $\cH$ defined as follows:
\begin{alignat*}{2}
&D(q_{\bR_+}) := \{ f \in \cH\,:\, f'|_{\bR_+} \in L^2(\bR_+, x \,dx)\},& \qquad &q_{\bR_+}(f) := \pi\int_0^{+\infty} x |f'(x)|^2 dx,\\
&D(q_B) := \{ f \in \cH\,:\, f'|_B \in L^2(B, (1-x^2)\,dx)\}, & &q_B(f) := \pi\int_{-1}^1 (1-x^2) |f'(x)|^2 dx.
\end{alignat*}
We will use the same symbols to denote the bilinear forms induced by these quadratic forms by polarization.

\begin{proposition}
The forms $q_{\bR_+}$, $q_B$ are closed.
\end{proposition}

\proof
We prove the statement for $q_B$, the proof for $q_{\bR_+}$ being completely analogous. Let $\{ f_n\} \subset D(q_B)$ be norm convergent to some $f \in \cH$ and such that $q_B(f_n-f_m) \to 0$ as $n, m \to +\infty$. Then $\{f_n'|_B\}$ converges in $L^2(B,(1-x^2)\,dx)$ to some $g \in L^2(B,(1-x^2)\,dx)$, and it is therefore sufficient to show that $f'|_B = g$. To this end, take $\varphi \in C^\infty_c(B)$ and compute
\[\begin{split}
\langle f', \varphi\rangle&=-\langle f, \varphi'\rangle = -2 \Im (f, \varphi) = \lim_{n \to +\infty}-2 \Im(f_n, \varphi) = \lim_{n \to +\infty}\int_{-1}^1 dx f'_n(x) \varphi(x)  \\
&= \lim_{n \to +\infty} \int_{-1}^1 dx (1-x^2) f'_n(x) \frac{\varphi(x)}{1-x^2} = \int_{-1}^1 dx (1-x^2)g(x)\frac{ \varphi(x)}{1-x^2} = \langle g,\varphi\rangle,
\end{split}\]
where in the second equality we used the definition of the scalar product in $\cH$ in terms of the Fourier transforms $\hat f$, $\hat \varphi$, in the third equality the norm convergence of $\{f_n\}$ to $f$, in the fourth equality again the definition of the scalar product and the fact that $f_n'$ is a function on $\supp \varphi \subset B$, and in the sixth equality the convergence of $\{f'_n|_B\}$ to $g$ in $L^2(B,(1-x^2)\,dx)$. The statement is thus proven. 
\eproof

Since $q_B$ is clearly non-negative, there exist a real self-adjoint operator $A_B$ on $\cH$ (thought as a real Hilbert space with the real part of the scalar product), such that
\[
q_B(f,g) = \Re(f,A_B g),\qquad f,g\in D(q_B),
\]
and a similar statement holds for $q_{\bR_+}$.

We are now ready to prove the main result of this appendix.

\begin{theorem}
The quadratic forms $q_{\bR_+}$, $q_B$ coincide with the entropy quadratic forms $S(\cdot | H(\bR_+))$, $S(\cdot | H(B))$ respectively.
\end{theorem}

\proof
Again we only prove the statement for $q_B$. We start by observing that clearly $C^\infty(\bS^1) \subset D(q_B)$. Moreover, given  $g \in C^\infty(\bS^1)$, let $h_a(x) := (a^2-x^2) g'(x)$, $a, x \in \bR$. Then, using once more~\eqref{eq:normL2bound} one checks that $\chi_{(-1+\varepsilon,1-\varepsilon)}h_{1-\varepsilon} \to \chi_Bh_1$ in $\cH$ as $\varepsilon \to 0$, and then mollifying $\chi_{(-1+\varepsilon,1-\varepsilon)}h_{1-\varepsilon}$ as in the proof of Prop.~\ref{prop:density} to obtain a sequence of $C^\infty_c(\bR)$ function, one concludes that $\chi_B h_1  \in H(B)$. Similarly, $(1-\chi_B)h_1 \in H(B^c)= H(B)'$, and therefore, from~\eqref{eq:logDelta},
\[
\cE_{H(B)} g = \imath P_{H(B)} \imath \log \Delta_{H(B)} g= \imath \chi_B h_1,
\]
and, for all $f \in D(q_B)$,
\[
q_B(f,g) = \pi \int_{-1}^1 dx\, (1-x^2) f'(x) g'(x) = \Re(f,\cE_{H(B)}g) = \Re(f,A_B g).
\]
Thus $A_B$ is a self-adjoint extension of $\cE_{H(B)}|_{C^\infty(\bS^1)}$, which is essentially self-adjoint by Cor.~\ref{cor:coreE}, so that, finally, $A_B = \cE_{H(B)}$ and $q_B = S(\cdot |H(B))$.
\eproof

We can thus write, for all $f \in \cH$,
\[
S(f|H(B)) = \pi\int_{-1}^1 (1-x^2) |f'(x)|^2 dx, \qquad S(f|H(\bR_+))=\pi\int_0^{+\infty} x |f'(x)|^2 dx,
\]
with the convention that the right hand sides are $+\infty$ if (the appropriate restriction of) $f'$ does not belong to $L^2(B,(1-x^2)\,dx)$ and $L^2(\bR_+,x\,dx)$ respectively.
\medskip

\noindent
{\bf Acknowledgments.} We thank Vincenzo Morinelli for several useful discussions on the subject of Appendix B. R.L. and G.M. acknowledge the Excellence Project 2023-2027 Mat-Mod@TOV 
awarded to the Department of Mathematics, University of Rome Tor Vergata. The work of G.M. is partly supported by INdAM-GNAMPA and by the University of Rome Tor Vergata funding OANGQS, CUP E83C25000580005.

\end{document}